\documentclass[draftclsnofoot,onecolumn,12pt,twoside]{IEEEtran}
\usepackage[noadjust]{cite}
\usepackage{amsmath,amssymb,amsfonts,amscd,amsthm,amsbsy}
\usepackage{graphicx}
\usepackage[english]{babel}
\usepackage{mathrsfs}
\usepackage{xcolor}
\usepackage{mathtools}
\usepackage{array}
\usepackage{blkarray}
\usepackage{caption}
\usepackage{subcaption}

\usepackage{tcolorbox}
\tcbuselibrary{breakable}
\newtcolorbox{mybox}[2][]{colbacktitle=white,colback=white,coltitle=black,title={#2},fonttitle=\bfseries,#1, left = 2mm, right = 2mm, breakable}

\newtheorem{theorem}{Theorem}
\newtheorem{claim}{Claim}
\newtheorem{lemma}{Lemma}

\newtheorem{corollary}{Corollary}
\newtheorem{remark}{Remark}

\newtheorem{example}{Example}
\newtheorem{definition}{Definition}

\newtheorem{observation}{Observation}

\usepackage{etoolbox}
\AtEndEnvironment{example}{\null\hfill\IEEEQED}%
\AtEndEnvironment{remark}{\null\hfill\IEEEQED}%

\newcommand{\ka}{\mathbb{\kappa}}
\newcommand{\Fb}{\mathbb{F}}

\newcommand{\xcal}{\chi}
\newcommand{\hcal}{\mathcal{H}}

\newcommand{\rogers}[1]{{\color{black} #1}}
\newcommand{\prasad}[1]{{\color{red} #1}}

\normalsize

\title{Pliable Index Coding via ~\\Conflict-Free Colorings of Hypergraphs}

  \author{Prasad Krishnan, Rogers Mathew, Subrahmanyam Kalyanasundaram}

\begin{document}

\maketitle
\thispagestyle{empty}
\begin{abstract}
In the pliable index coding (PICOD) problem, a server is to serve multiple clients, each of which possesses a unique subset of the complete message set as side information and requests a new message which it does not have. The goal of the server is to do this using as few transmissions as possible. This work presents a hypergraph coloring approach to the scalar PICOD problem. A \textit{conflict-free coloring} of a hypergraph is known from literature as an assignment of colors to its vertices so that each hyperedge of the graph contains one uniquely colored vertex. For a given PICOD problem represented by a hypergraph consisting of messages as vertices and request-sets as hyperedges, we present achievable PICOD schemes using conflict-free colorings of the PICOD hypergraph. Various graph theoretic parameters arising out of such colorings (and some new coloring variants) then give a number of upper bounds on the optimal PICOD length, which we study in this work. 

Suppose  the PICOD hypergraph has $m$ vertices and $n$ hyperedges,
    where every hyperedge overlaps with at most $\Gamma$ other hyperedges.
We show easy to implement randomized algorithms for the following:
\begin{itemize}
    \item For the single request case, we give a PICOD of length  $O(\log^2\Gamma)$.
    This result improves over known achievability  results \cite{brahma_fragouli_PICOD,PolyTime_PICOD} for some parameter ranges.  
    \item For the $t$-request case,  we give an MDS code of length $\max(O(\log \Gamma \log m), O(t \log m))$. Further if the hyperedges (request sets) are sufficiently large, we give a PICOD 
    of the same length as above, which is not
    based on MDS construction. In general, this gives an improvement over the results of \cite{PolyTime_PICOD}. Our codes are of near-optimal length (up to a multiplicative 
    factor of $\log t$).
\end{itemize}
\end{abstract}
\let\thefootnote\relax\footnotetext{
 Dr.\ Krishnan is with the Signal Processing \& Communications Research Center, International Institute of Information Technology Hyderabad, India, email:prasad.krishnan@iiit.ac.in. Dr.\ Mathew and Dr.\ Kalyanasundaram are with the Department of Computer Science and Engineering, Indian Institute of Technology Hyderabad, email: \{rogers,subruk\}@iith.ac.in.~\\ 
 Acknowledgements: The first author acknowledges support from SERB-DST projects MTR/2017/000827 and CRG/2019/005572. The second author was supported by a grant from the Science and Engineering Research Board, Department of Science and Technology, Govt. of India (project number: MTR/2019/000550).  
 The third author acknowledges SERB-DST (MTR/2020/000497) for supporting this research.   
}
\section{introduction}
The Index Coding problem introduced by Birk and Kol in \cite{BiK} consists of a system with a server containing $m$ messages and $n$ receivers connected by a broadcast channel. Each receiver has a subset of the messages at the server as \textit{side-information} and demands a particular new message. The goal of the index coding problem is to design a transmission scheme at the server which uses minimum number of transmissions to serve all receivers, also called the length of the index code. The index coding problem is a canonical problem in information theory and has been addressed by a variety of techniques, including graph theory \cite{BBJK,LocalColoringSDL}, linear programming \cite{BKL2}, interference alignment \cite{MCJ}, etc. 

A variant of the index coding problem, called pliable index coding (PICOD), was introduced by Brahma and Fragouli in \cite{brahma_fragouli_PICOD}. The pliable index coding problem relaxes the index coding setup, such that each receiver requests \textit{any} message which is not present in its side-information (i.e., any message from its \textit{request-set}). It was shown in \cite{brahma_fragouli_PICOD} that finding the optimal length of a PICOD problem is NP-hard in general. However the existence of a code with length $O(\min\{\log m(1+\log^+(\frac{n}{\log m})), m, n\})$ was proved using a probabilistic argument (where $\log^+(x)=\max\{0,\log(x)\})$. When $m=n^{\delta}$ for some constant $\delta>0$, this means that $O(\log^2 n)$ is sufficient. Some algorithms for designing pliable index codes based on greedy and set-cover techniques were also presented and compared in \cite{brahma_fragouli_PICOD}. In \cite{PolyTime_PICOD}, a polynomial-time algorithm was presented for general PICOD problems which achieves a length $O(\log^2 n)$. Thus, unlike the index coding problem which has instances for which the required length can be $\Theta(n)$ (for instance, the directed $n$-cycle problem \cite{BBJK}), much fewer transmissions are sufficient in general for PICOD instances. For several special classes of PICOD problems, distinguished by the structure of the side-information or request-sets of the receivers, achievability and converse results were presented in \cite{TightITConverse_LiuTun_TIT,Sasi_Rajan_CodeConstrPICOD}. Converse techniques were further developed in \cite{Ong_Vellambi_Kliewer_AbsentRx1,ong2019improved_absRx2}, using which the optimal lengths of specific classes of PICOD problems were obtained. Other extensions of  PICOD such as vector pliable index codes \cite{PolyTime_PICOD}, multiple requests \cite{brahma_fragouli_PICOD,PolyTime_PICOD}, secure pliable index codes \cite{SecureDecPICOD} and 
decentralized pliable index codes \cite{Decen_PICOD} have  been studied recently. Pliable index coding has also been proposed for efficient data exchange in real-world applications, such as in the data shuffling phase of distributed computing \cite{PICODdatashuffling}.

In this work, we present a graph coloring approach to pliable index coding.  A \textit{conflict-free coloring} of a hypergraph is an assignment of labels to its vertices so that each hyperedge of the hypergraph contains at least one vertex which has a label distinct from others. Conflict-free colorings were introduced by Even et al. in \cite{even2003conflict}, motivated by a problem of frequency assignment in wireless communications. Since then, it has been extensively studied in the context of general hypergraphs, hypergraphs induced by neighborhoods in graphs, hypergraphs induced by simple paths in a graph, hypergraphs that naturally arise in geometry, etc. See \cite{smorodinsky2013conflict} for a survey on conflict-free colorings.

Any PICOD problem can equivalently be represented using a hypergraph 
with the vertices representing the messages, and the request-sets as hyperedges. We show that conflict-free colorings (and its variants) of this hypergraph 
give rise to achievability schemes for PICOD. 

 Our specific contributions and organization of this paper are as follows. 

\begin{itemize} 
    \item We  briefly review the PICOD problem setup in Section \ref{sec:PICODreview} and conflict-free colorings in Section \ref{sec:conflictfreecolorings}.  In Section \ref{subsec:relationshipPICCF}, we show that the optimal length of a PICOD problem
    is at most the conflict-free chromatic number of the hypergraph corresponding to the PICOD
    problem (Lemma \ref{lemma:upperboundchromnumber}). The conflict-free chromatic number of a hypergraph is the smallest number of colors
    required to conflict-free color the hypergraph. 

    \item In Section \ref{subsec:CFcoveringsandPICOD}, we define the notion of \textit{conflict-free collection of colorings} of hypergraphs. We call the corresponding chromatic number the \emph{conflict-free covering number}. This notion gives a better upper bound 
    than what is given by conflict-free coloring (Lemma \ref{lemma:upperboundCFcover}). Using conflict-free collections, we show in Section \ref{subsec:logsqrgammaresult} that $O(\log^2 \Gamma)$ transmissions suffice for scalar PICOD schemes, where $\Gamma$ refers to the maximum number of hyperedges any single hyperedge overlaps with (Theorem \ref{thm_CF_cover}). 
    This result improves over known achievability  results \cite{brahma_fragouli_PICOD,PolyTime_PICOD} for some parameter ranges.  
    Our proof for Theorem \ref{thm_CF_cover} uses a  probabilistic argument, but this can be converted into a 
    deterministic polynomial time algorithm using known techniques (see the discussion towards the end of Section \ref{subsec:logsqrgammaresult}). 
    \item In Section \ref{sec:local}, we define a parameter called \textit{local} conflict-free chromatic number and show a PICOD scheme whose length is upper bounded by this parameter (Theorem \ref{theoremMDS}). As 
    the local conflict-free chromatic number can be smaller than the conflict-free chromatic number, this improves the upper bound in Section \ref{subsec:relationshipPICCF}. Using conflict-free collection of colorings, in Section \ref{subsec:localcovering_PICOD}, we also generalize the covering number of Section \ref{subsec:CFcoveringsandPICOD} to its local variant and show the corresponding scheme (Theorem \ref{thm:localcoverbound}). 

    \item In Section \ref{sec:trequest}, we study the $t$-request PICOD problem where each receiver wants
    $t$ messages from its request set. This corresponds to a conflict-free coloring of 
    the hypergraph, where each hyperedge sees $t$ colors exactly once. The corresponding chromatic numer is called 
    the \emph{$t$-strong conflict-free chromatic number}. Analogous to conflict-free covering
    number, local conflict-free chromatic number and local conflict-free covering number, we define
    ``$t$-strong'' variants of each of these parameters. We observe that each of these parameters
    upper bound the length of an optimal $t$-request PICOD.
    
    \item Let $\hcal$ be the $t$-request PICOD hypergraph with $m$ vertices and $n$ hyperedges,
    where every hyperedge overlaps with at most $\Gamma$ other hyperedges. In Section \ref{sec:lambdaupperboundtstrong},
    we give a simple randomized algorithm to construct an MDS code of length $\max(O(\log \Gamma \log m), O(t \log m))$. In Section \ref{sec:alphaupperboundtstrong}, we give another randomized
    approach to construct a code of length  
    
    $\max(O(\log \Gamma \log m), O(t \log m))$, provided 
    every request set (or hyperedge) is sufficiently large.
    In general, this gives an improvement over the results of \cite{PolyTime_PICOD} where a 
    PICOD of length $O(t\log n+\log^2 n)$ was shown. 
    In Section \ref{sec:tstrongtight}, we show that the lengths of the codes yielded by the above constructions are asymptotically tight up to a multiplicative factor of $\log t$.

     \item In Section \ref{sec:k-folddefinition}, we define a generalization 
     of conflict-free coloring called \textit{$k$-fold conflict-free coloring}. This corresponds to the $k$-vector pliable index codes. 
     
\end{itemize}

\textit{Notations:} Let $[n]\triangleq \{1,\hdots,n\}$ for a positive integer $n$. For sets $A,B$ we denote by $A\backslash B$ the set of elements in $A$ but not in $B$. We abuse notation to denote $A\backslash \{b\}$ as $A\backslash b$. The set of $k$-sized subsets of any set $A$ is denoted by $\binom{A}{k}.$ The span of a set of vectors $U$ is denoted by $span(U)$. The dimension of a subspace $W$ is denoted by $\dim(W)$. Unless mentioned explicitly, all logarithms in the paper are to the base $e$. The empty is denoted by $\emptyset$. A hypergraph $\hcal$ is a pair of sets $(V,{\cal E})$ where the set $V$ is called the set of vertices of $\hcal$ (also denoted by $V(\hcal)$) and $\cal E$ is a collection of subsets of $V$ called the set of hyperedges (sometimes referred to as simply edges) of $\hcal$ (also denoted by ${\cal E}(\hcal)$). Given a collection of hypergraphs $\hcal_p:p\in[P]$, we define their union as $\hcal=\cup_{p\in[P]}\hcal_p$ where $V(\hcal)=\cup_{p\in[P]}V(\hcal_p)$, and ${\cal E}(\hcal)=\cup_{p\in[P]}{\cal E}(\hcal_p)$.
\section{Pliable Index Coding Problem}
\label{sec:PICODreview}
We briefly review the pliable index coding problem, introduced in \cite{brahma_fragouli_PICOD}. Consider a communication setup defined as follows. There are $m$ messages denoted by $\{x_i:i\in[m]\}$ where $x_i$ lies in some finite alphabet ${\cal A}$. These $m$ messages are available at a server. Consider $n$ receivers indexed by $[n]$. Assume that there is a noise-free broadcast channel between the server and the receivers. Each receiver $r$ has some subset of messages available apriori, as \textit{side-information}. The set of indices
of the symbols available at receiver $r$ is denoted $S_r$, and those that are not available is denoted 
$I_r= [m] \setminus S_r$. 
We call $\{x_i:i\in I_r\}$ as the \textit{request-set} of receiver $r$. 
The demand at receiver $r$ is fulfilled if it receives any symbol in its request set from the server.
The messages indexed by $[m]$, the receivers indexed by $[n]$, and the request-sets ${\mathfrak I}\triangleq \{I_r:r\in[n]\}$ together define a $(n,m,{\mathfrak I})$-\textit{pliable index coding problem (PICOD problem).}  We assume that $|I_r|\geq 1, \forall r$, as any receiver with $|I_r|=0$ can be removed from the problem description as it has all the symbols. Consider a hypergraph $\cal{H}$ with vertex set $V=[m]$ and edge set $\mathfrak{I}=\{I_r:r\in[n]\}.$ This hypergraph  captures the PICOD problem. 



A \textit{pliable index code} (PIC) consists of a collection of (a) an encoding function at the server which encodes the $m$ messages to an $\ell$-length codeword, denoted by $\phi:{\cal A}^m\rightarrow {\cal A}^\ell$ and (b) decoding functions $\{\psi_r:r\in [n]\}$ where $\psi_r:{\cal A}^\ell\times {\cal A}^{|S_r|}\rightarrow {\cal A}$ denotes the decoding function at receiver $r$ such that \[
\psi_r\left(\phi(\{x_i:i\in[m]\}),\{x_i:i\in S_r\}\right)=x_d,~\text{for some}~d\in I_r.
\]
The quantity $\ell$ is called the \textit{length} of the PIC.  It is of interest to design pliable index codes of small length.

In this work, we assume ${\cal A}={\mathbb F}^k$ for some finite field $\mathbb F$ and integer $k\geq 1$. We refer to these codes as \textit{$k$-vector PICs}, while the $k=1$ case is also called \textit{scalar PIC}. We focus on linear PICs, i.e., one in which the encoding and decoding functions are linear. In that case, the encoder $\phi$ is represented by a $\ell\times mk$ matrix (denoted by $G$)  such that 
$\phi(\{x_i:i\in[m]\})=G\boldsymbol{x}^T,$ where $\boldsymbol{x}=(x_{1,1},\hdots,x_{1,k},\hdots,x_{m,1},\hdots,x_{m,k}).$ 
For the PICOD problem given by hypergraph $\hcal$, the smallest $\ell$ for which there is a linear $k$-vector PIC
(over some field $\mathbb F$) is denoted by $\ell_k^*({\hcal})$.

The following definition and lemma (which is proved in \cite{PolyTime_PICOD}) describe when $G$ can lead to correct decoding at the receivers.
\begin{definition}
\label{DefnGeneratorForEdge}
For an $(n,m,{\mathfrak I})$-PICOD problem, a matrix $G$ with $mk$ columns indexed as $G_{i,j}:i\in[m],j\in[k],$ is said to \textit{satisfy} receiver $r\in[n],$ if the following property (P) is satisfied by $G$.
\begin{itemize}
    \item[\textbf{(P)}] There exists some $d\in I_r$ such that $\dim(span(\{G_{d,j}:j\in[k]\}))=k$ and 
    \begin{align*}
span(\{G_{d,j}&:j\in[k]\}) \bigcap span \left(\{G_{i,j}:~\forall i\in I_r\backslash d, ~ j\in[k]\}\right)=\{\boldsymbol{0}\}.
    \end{align*}
\end{itemize}
\end{definition}
\begin{lemma}[\cite{PolyTime_PICOD} Lemmas 1 and 6]
\label{lemmaindependence}
A matrix $G$ with $mk$ columns is the encoder of a PIC for an $(n,m,{\mathfrak I})$-PICOD problem if and only if the property (P) of Definition \ref{DefnGeneratorForEdge} is true for each receiver $r\in[n]$. 
\end{lemma}
\rogers{Lemma \ref{lemma:concatenatedGeneratorsPIC} below} is useful to prove achievability results for PICOD problems in this work.  
\begin{lemma}
\label{lemma:concatenatedGeneratorsPIC}
For an $(n,m,{\mathfrak I})$-PICOD problem, let $\{G^p:p\in[P]\}$ denote a collection of matrices, where $G^p$ is of size $L_p\times mk,$ such that for each $r\in[n],$ there exists some matrix $G^p$ which satisfies receiver $r$. Then the matrix $G=\begin{bmatrix}
G^1\\\vdots\\G^P\end{bmatrix}$ of size $(\sum_{p\in[P]}L_p)\times mk$ is the encoder of a PIC for the given PICOD problem. 
\end{lemma}
\begin{IEEEproof}
For each $r\in[n]$, there exists some matrix $G^p$ such that Property (P) holds for $r$ (with respect to some $d\in I_r$). By simple linear algebra, we see that the matrix $G$ too must satisfy property (P) for receiver $r$ (with respect to $d\in I_r$), and hence satisfies $r$. Applying Lemma \ref{lemmaindependence}, the proof is complete. 
\end{IEEEproof}

\section{Scalar PICs arising from Conflict-free Colorings}
\label{sec:conflictfreecolorings}
Firstly, we review the \rogers{definition of} conflict-free colorings of a hypergraph and discuss some existing results. 
In Subsection \ref{subsec:relationshipPICCF}, we show that a conflict-free coloring of the hypergraph ${\hcal}(V=[m],\mathfrak{I})$ gives a scalar linear PIC scheme for the PICOD problem given by ${\hcal}$, with length equal to the number of colors. We then define in Subsection \ref{subsec:CFcoveringsandPICOD} the notion of a conflict-free collection of colorings of $\hcal$, and show that such a collection also leads to achievable PIC schemes. This yields tighter upper bounds on $\ell_1^*({\hcal})$, in general. In Subsection \ref{subsec:logsqrgammaresult}, we show that the PICOD problem $\hcal$ has a PIC scheme with length $O(\log^2\Gamma)$, where $\Gamma$ is an edge-overlap parameter associated with $\hcal$. This upper bound gives an order-wise improvement over the $O(\log^2n)$ bound shown in \cite{PolyTime_PICOD}. 

Let $\mathcal{H}=(V,\mathcal{E})$ be a hypergraph. Let $C:V \rightarrow [L]$ be a coloring of $V$, where $L$ is a positive integer. 
Consider a hyperedge $E \in \mathcal{E}$. We say $C$ is a \emph{conflict-free coloring for the hyperedge $E$} if there is a vertex $v \in E$ such that $C(v) \neq C(u),~\forall u \in E \setminus \{v\}$. That is, in such a coloring, $E$ contains a vertex whose color is distinct from that of every other vertex in $E$. We say $C$ is a \emph{conflict-free coloring of the hypergraph $\mathcal{H}$} if $C$ is a conflict-free coloring for every $E\in {\cal E}$.  
The \emph{conflict-free chromatic number} of $\mathcal{H}$, denoted by $\xcal_{CF}(\mathcal{H})$,  is the minimum $L$ such that there is a  conflict-free coloring $C:V \rightarrow [L]$ of $\mathcal{H}$. 
The following theorem on conflict-free coloring on hypergraphs is due to  Pach and Tardos \cite{Pach2009}, which we shall use to obtain one of our main results (Theorem \ref{thm_CF_cover} and Corollary \ref{corollary_achievability}) in Subsection \ref{subsec:logsqrgammaresult}.
\begin{theorem}[Theorem 1.2 in \cite{Pach2009}]\label{thm_pach_main}
For any positive integers $t$ and $\Gamma$, the conflict-free chromatic number of any hypergraph in which each edge is of size at least $2t-1$ and each edge intersects at most $\Gamma$ others is $O(t\Gamma^{1/t}\log \Gamma)$. There is a randomized polynomial time algorithm to find such a coloring.   
\end{theorem}

\subsection{Relationship of PIC to Conflict-free Coloring}
\label{subsec:relationshipPICCF}
In this subsection, we show that a conflict-free coloring of the hypergraph ${\hcal}(V=[m],\mathfrak{I})$ gives a scalar PIC scheme for the PICOD problem given by ${\hcal}$. To do this, we define the following matrix associated with a conflict-free coloring of ${\hcal}.$
\begin{definition}[Indicator Matrix associated with a coloring]
\label{DefngenmatrixAssocWithaColoring}
Let $C:V\rightarrow [L]$ denote a coloring of ${\hcal}(V = [m],{\mathfrak{I}})$, where $C(i)$ denotes the color assigned to the vertex $i\in[m]$. Consider a standard basis of the $L$-dimensional vector space over ${\mathbb F}$, denoted by $\{e_1,\hdots,e_L\}$. Now consider the $L\times m$ matrix $G$ (with columns indexed as $\{G_{i}:i\in[m]\}$) constructed as follows. 
\begin{itemize}
    \item For each $i\in[m]$, column $G_{i}$ of $G$ is fixed to be $e_{C(i)}$.
\end{itemize}
We call $G$ as the \textit{indicator matrix associated with the coloring $C$. }
\end{definition}
Using the indicator matrix associated with a conflict-free coloring of ${\hcal},$ we shall prove our first bound on $\ell^*_1({\hcal}).$
\begin{lemma}
\label{lemma:upperboundchromnumber}
$\ell^*_1({\hcal})\leq \xcal_{CF}({\hcal})$.
\end{lemma}
\begin{IEEEproof}
Let $C:V\rightarrow [L]$ denote a conflict-free coloring of ${\hcal}$. We first show that there exists an $L$-length scalar linear PIC for the  problem defined by ${\hcal}$.
Let $G$ denote the indicator matrix associated with the coloring $C$ as defined in Definition \ref{DefngenmatrixAssocWithaColoring}. We show that $G$ satisfies Lemma \ref{lemmaindependence} and hence is a valid encoder for a linear PIC. 

In any conflict-free coloring of ${\hcal}$, every edge $I_r$ of ${\hcal}$ has a vertex $d$ such that $C(d)\neq C(i), \forall i\in I_r\setminus d$.  Then, clearly, $e_{C(d)}\neq e_{C_{(i)}},$ for any $i\in I_r\backslash d.$ This also means $span(\{e_{C_{(d)}}\})\cap span(\{e_{C_{(i)}}: i\in I_r\backslash d\})=\{\boldsymbol{0}\}$, as the vectors $\{e_1,\hdots,e_L\}$ are basis vectors. Further, $e_{C_{(d)}}$ spans a one dimensional space. Thus, $G$ satisfies every receiver $r$ and is a valid encoder by Lemma \ref{lemmaindependence}. Note that the length of the code is exactly $L$. By definition of $\xcal_{CF}({\hcal}),$ the proof is complete. 
\end{IEEEproof}
\begin{example}
Consider the PICOD problem represented by the hypergraph ${\hcal}$ with vertex set $V=\{1,\hdots,8\}$ and edge set \begin{align*}{\cal E}=\{
\{1,2,4,6\},\{1,2,3,5\},\{2,3,4,7\},\{1,3,4,8\},\\~~~\{2,5,6,7\},\{1,5,6,8\},\{3,5,7,8\},\{4,6,7,8\}
\}. 
\end{align*}
Consider a coloring $C$ which assigns color $1$ to vertices $\{1,2,3,4\}$ and color $2$ to vertices $\{5,6,7,8\}.$ Note that this is a valid conflict-free coloring of ${\hcal}$. The indicator matrix associated with this coloring is given by
\begin{align*}
    G=\begin{bmatrix}
    1&1&1&1&0&0&0&0\\
    0&0&0&0&1&1&1&1
    \end{bmatrix}.
\end{align*}
It can be checked that the above matrix satisfies the condition in Lemma \ref{lemmaindependence} for the PICOD problem defined by $\hcal.$
\end{example}
\subsection{Conflict-free coverings and PICOD} \label{subsec:CFcoveringsandPICOD} 
In the following discussion, we define a new parameter called the conflict-free covering number, which in general improves upon the upper bound on the optimal length as given in Lemma \ref{lemma:upperboundchromnumber}. 
\begin{definition}[\textit{Conflict-free collection, conflict-free covering number}]
\label{defn:cfcoveringnumber}
Let $\mathcal{H}=(V,\mathcal{E})$ be a hypergraph. Let $\mathfrak{C} = \{C^1, \ldots , C^P\}$ where each $C^p:V\rightarrow [L_p]$ be colorings of the hypergraph $\hcal$. We say $\mathfrak{C}$ is a \emph{conflict-free collection of colorings of $\mathcal{H}$}, if the following condition holds: For every $E \in \mathcal E$, there is $p \in [P]$ such that E sees
some color exactly once under the coloring $C_p$.

The quantity $$\alpha_{CF}(\mathcal{H})\triangleq \min_{\mathfrak{C}}\sum_{p=1}^P L_p,$$ representing the minimum sum $\sum_{p=1}^PL_p$ over all possible collections $\mathfrak{C}$ (over all $P$) as defined above, is called the \emph{conflict-free covering number of $\mathcal{H}$}.

\end{definition}
In the following, we show that the parameter $\alpha_{CF}(\hcal)$ is sandwiched between functions of $\xcal_{CF}(\hcal)$.
\begin{lemma}
\label{lemma:generalhypergraph_alphabounds}
Let $\hcal$ be a hypergraph with $\xcal_{CF}(\hcal)=\xcal$ and let $r$ be the smallest integer such that in any conflict-free coloring of $\hcal$ using $\xcal$ colors, the vertices in any hyperedge are colored with at most $r$ colors. If $r=1,$ then $\xcal=\alpha_{CF}(\hcal)=1$ and if $r\geq 2$, then $\log_2(\xcal) \leq \alpha_{CF}(\hcal)\leq \min(\xcal,\frac{r^{r+2}}{r!}\log_e(\xcal)).$ 
\end{lemma}
\begin{IEEEproof}
If $r=1$, it is easy to see that the lemma holds. So we assume $r\geq 2$ throughout. 

Consider a conflict-free collection of $\hcal$ with $P$ colorings with $\alpha_{CF}(\hcal)$ total colors. Let $V=[m]$ be the set of vertices of $\hcal$. Consider the $\alpha_{CF}(\hcal)\times |V|$ matrix $G=\begin{bmatrix}G^1\\\vdots\\G^P\end{bmatrix},$ where $G^p$ represents the indicator matrix of the $p^{th}$ coloring in the collection. Let $\{g_1,\hdots,g_L\}$ denote the set of all distinct columns of $G$. Thus, we must have that $\alpha_{CF}(\hcal)\geq \log_2(L).$ Now, consider a coloring $C$ of the vertices of $\hcal$ with elements of $[L]$, where a vertex $i\in V$ gets label $\ell\in[L]$ if the $i^{th}$ column of $G$ is $g_{\ell}.$ By construction of $G$, $C$ is a conflict-free coloring of $\hcal$ and thus $L\geq \xcal.$ Thus we see that $\alpha_{CF}(\hcal)\geq \log_2(\xcal).$

The upper bound $\alpha_{CF}(\hcal)\leq \xcal$ follows as any conflict-free coloring of ${\hcal}$ also gives a conflict-free collection containing only the same coloring. We now show the other upper bound. Let $C:V\rightarrow [\xcal]$ be any coloring with $\xcal$ colors so that vertices in any hyperedge are colored with at most $r$ colors. Firstly, for $P=\frac{r^{r+1}\log_e\xcal}{r!}$, we show that there exists a $P\times \xcal$ matrix with entries from $[r]$ such that any $r$ columns of this matrix contain at least one row with $r$ distinct entries. Using this matrix, we construct a conflict-free collection of colorings with $\frac{r^{r+2}}{r!}\log_e(\xcal)$ colors, which will complete the proof. 

We now show the existence of this desired matrix. Let each entry of a random $P\times \xcal$ matrix $M$ be i.i.d and drawn uniformly at random from $[r].$ The probability that any particular $r$-subset $R\subset \xcal$ of columns contains $r$ distinct colors in the row $i$ is given by $\frac{r!}{r^r}.$ Thus, the probability $q$ that at least one $r$-subset of columns of $M$ does not contain distinct entries in any row, is given by 
\begin{align*}
    q\leq \binom{\xcal}{r}\left(1-\frac{r!}{r^r}\right)^P\stackrel{(a)}{<} \xcal^re^{-P\frac{r!}{r^r}}\leq 1,
\end{align*}
where $(a)$ follows as $\binom{\xcal}{r}<\xcal^r$ for $r\geq 2$, and since $1+x\leq e^x, \forall x$. This means that there is at least one matrix $M$ (say $M_{cf}$) of size $P\times \xcal$ with entries from $[r],$ such that each $r$-subset of columns contains distinct entries in some row. 

Now, using $M_{cf}$, we obtain a collection of $P$ colorings of $\hcal$ in the following way. With respect to the $p^{th}$ row of $M_{cf}$, we define a coloring $C^p:V\rightarrow [r]$ such that for each $i\in V,$ $C^p(i)$ is equal to the entry in the $j^{th}$ column of $M_{cf},$ where $j$ is the color assigned to the vertex $i$ in $C$. By the property of $M_{cf}$ and the coloring $C$ chosen, it can be verified that this collection $\{C^p:p\in[P]\}$ will be a conflict-free collection of colorings of $\xcal$. The collection uses $r$ colors per coloring and thus totally there are $rp$ colors being used. 
\end{IEEEproof}
Fig. \ref{fig2} gives an example hypergraph for which $\alpha_{CF}({\hcal})<\xcal_{CF}({\hcal})$. 
\begin{figure*} 
    \centering
  \subfloat[\label{2a}]{%
       \includegraphics[width=0.18\linewidth]{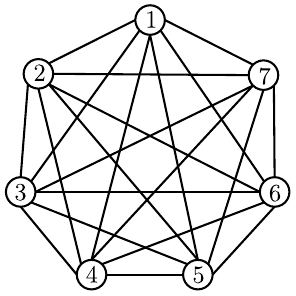}}
       \hfill
  \subfloat[\label{2b}]{%
        \includegraphics[width=0.18\linewidth]{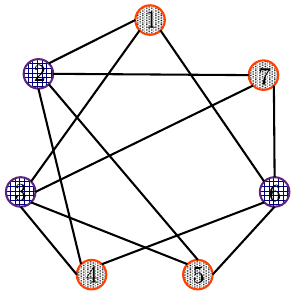}}
       \hfill
  \subfloat[\label{2c}]{%
        \includegraphics[width=0.18\linewidth]{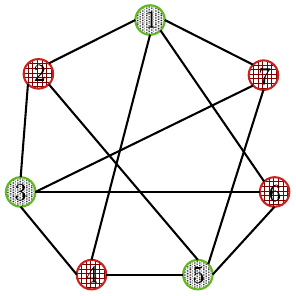}}
    \hfill
  \subfloat[\label{2d}]{%
        \includegraphics[width=0.18\linewidth]{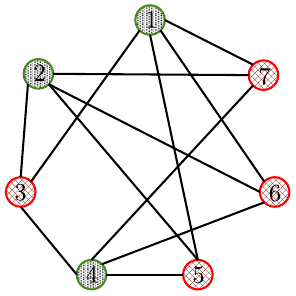}}
  \caption{\small \textcolor{black}{Figure (a) shows the hypergraph $K_7$ which is the complete graph on $7$ vertices, $\{1, 2, 3, 4, 5, 6, 7\}$ containing all the 2-sized subsets as hyperedges.  It requires $7$ colors for a conflict-free coloring, thus $\xcal_{CF}(K_7)=7.$ The three figures (b), (c) and (d) depict a collection of $1$-fold conflict-free colorings, each figure corresponding to one coloring using $2$ colors. Note that only those edges satisfied by the coloring are represented in (b), (c), and (d). In (b), the two color classes are $\{1,4,5,7\}$ and $\{2,3,6\}$. In (c), they are $\{2,4,6,7\}$ and $\{1,3,5\}$ and in (d) they are $\{1,2,4\}$ and $\{3,5,6,7\}$. It can be checked that each edge of $K_7$ is conflict-free in at least one of these colorings. Thus $\alpha_{CF}({\hcal})\leq 6<\xcal_{CF}({\hcal}).$}}
  \vspace{-0.4cm}
  \label{fig2}
\end{figure*}
The following lemma gives a class of hypergraphs for which the separation between parameters $\alpha_{CF}({\hcal})$ and $\xcal_{CF}({\hcal})$ can be quite large.  Fig. \ref{fig2} gives an example hypergraph from this class and illustrates the lemma. We observe, however, that the minimum of the two upper bounds as in Lemma \ref{lemma:generalhypergraph_alphabounds} is still (asymptotically) tight for this class of hypergraphs.
\begin{lemma}
\label{lemma_2uniformparameters1}
There exist a hypergraph $\hcal$ with $n$ hyperedges for which $\alpha_{CF}({\hcal})=\Theta(\log n)$ while $\xcal_{CF}({\hcal})=\Theta(\sqrt{n})$.
\end{lemma}
\begin{IEEEproof}
Consider the $2$-uniform hypergraph with $m$ vertices and all the $2$-sized subsets of $[m]$ as hyperedges. \rogers{Thus $n=\binom{m}{2}$.} It is easy to see that any conflict-free coloring of this graph requires $m=\Theta(\sqrt{n})$ colors.

\rogers{
Let us now turn our attention to $\alpha_{CF}({\hcal}).$ Consider a conflict-free collection of $P$ colorings $C^p:p\in[P]$, for some integer $P$, each with number of colors $L_p,$ such that $\sum_{p\in P}L_p = \alpha_{CF}({\hcal})$. For each $p \in [P]$, let $G^p$ be the indicator matrix associated with the coloring $C^p$. Consider the $(\alpha_{CF}({\hcal}) \times m)$ binary matrix $G=\begin{bmatrix} G^1\\\vdots\\G^P\end{bmatrix}$.  By the construction in Lemma \ref{lemma:upperboundCFcover}, this is a valid encoder for a PIC of ${\hcal}$. Since every $2$-sized subset of $[m]$ is a hyperedge in $\mathcal{H}$, no two columns of $G$ are thus identical. Thus, $\alpha_{CF}({\hcal}) \geq \log_2 m$. In order to prove an upper bound for $\alpha_{CF}({\hcal})$, let $P := \lceil \log_2 m \rceil$.  Given any assignment of distinct $P$-bit binary vectors to the elements of $[m]$, one can construct a conflict-free collection of $P$ colorings of ${\hcal}$ given as $C^p:[m]\rightarrow \{c_p^0,c_p^1\}$ for $p\in[P]$, where  $C^p(j) = c_p^0$ (or $= c_p^1$) if the $p$-th bit in the binary vector associated with $j$ is $0$ (respectively, $1$). Thus, $\alpha_{CF}({\hcal}) \leq 2\lceil \log_2 m \rceil$.  
}\end{IEEEproof}

\subsection{An upper bound on $\ell^*_1({\hcal})$ via $\alpha_{CF}({\hcal})$}
\label{subsec:logsqrgammaresult}
In the remainder of this section, we  prove a main result of this work, which is a new upper bound (Theorem \ref{thm_CF_cover} and Corollary \ref{corollary_achievability}) on $\alpha_{CF}({\hcal})$ (and thus on the optimal linear scalar PIC length) based on a readily computable parameter associated with the hypergraph ${\hcal}$. Towards that end, we first show that the optimal length of PIC for ${\hcal}$ is bounded by $\alpha_{CF}({\hcal})$, thus improving the bound in Lemma \ref{lemma:upperboundchromnumber}. 
\begin{lemma}
\label{lemma:upperboundCFcover}
$\ell_1^*({\hcal})\leq \alpha_{CF}(\mathcal{H}).$
\end{lemma}
\begin{IEEEproof}
Let ${\mathfrak C}=\{C^p:p\in[P]\}$ be a conflict-free collection of colorings of ${\hcal}(V=[m],\mathfrak{I}),$ where $C^p:V\rightarrow [L_p]$. We first show a PIC for ${\hcal}$ with length $\sum_{p\in[P]}L_p.$ The proof then follows by definition of  $\alpha_{CF}({\hcal})$. 

Let $G^p:p\in[P]$ denote the indicator matrices as defined in Definition \ref{DefngenmatrixAssocWithaColoring} associated with the colorings $C^p:p\in[P]$ respectively. By definition of the conflict-free collection, for each $I_r\in\mathfrak{I}$ we have by arguments similar to the proof of Lemma \ref{lemma:upperboundchromnumber}, that there is some $G^p$ which satisfies receiver $r$.  Then, by Lemma \ref{lemma:concatenatedGeneratorsPIC}, the matrix $G=\begin{bmatrix}
G^1\\\vdots\\G^P\end{bmatrix}$ is a valid encoder of a $k$-vector PIC to the given PICOD problem of length $\sum_{p\in[P]}L_p$. 
\end{IEEEproof}

The following observation is also needed to show our new upper bound. 
\begin{observation}
\label{obv_union}
Let $\mathcal{H} = \mathcal{H}_1 \cup \mathcal{H}_2$. Then, $\alpha_{CF}(\mathcal{H}) \leq \alpha_{CF}(\mathcal{H}_1) + \alpha_{CF}(\mathcal{H}_2)$. 
\end{observation}
We are now ready to prove the main result in this section, which is a new upper bound on the optimal length of scalar PIC schemes. 
\begin{theorem}
\label{thm_CF_cover}
Let $\mathcal{H} = (V, \mathcal{E})$ be a hypergraph where every hyperedge intersects with at most $\Gamma$ other hyperedges, for any $\Gamma > e$, the base of the natural logarithm. Then, $\alpha_{CF}(\mathcal{H}) = O(\log^2\Gamma)$. 
\end{theorem}
\begin{IEEEproof}
Let $\ka := 2\log(\Gamma) - 1$. Let $\mathcal{G} = (V, \mathcal{E}_G)$ be a hypergraph defined on the vertex set $V$ with $\mathcal{E}_G = \{E \in \mathcal{E}~:~|E| \geq \ka\}$. From Theorem \ref{thm_pach_main} and Lemma \ref{lemma:generalhypergraph_alphabounds} ($\alpha_{CF}(\mathcal{G}) \leq \xcal_{CF}(\mathcal G)$), we know that $\alpha_{CF}(\mathcal{G}) = O(\log^2\Gamma)$. Let $P := \lceil \log \ka \rceil$. For $0 \leq i \leq P$, let $\mathcal{H}_i = (V, \mathcal{E}_i)$, where $\mathcal{E}_i = \{E \in \mathcal{E}~:~\frac{k_i}{2} \leq |E| < k_i\}$ and $k_i = \frac{\ka}{2^i}$. 
Clearly, $\mathcal{H} = \mathcal{G} \cup \mathcal{H}_0 \cup \mathcal{H}_1 \cup \cdots \cup \mathcal{H}_P$. We shall use the following claim whose proof uses the Lov{\'a}sz Local Lemma \cite{lovaszlocallemma} and is relegated to Appendix \ref{appendix_proof_claim}.

\begin{claim}
\label{lem_Hi} $\alpha_{CF}({\hcal}_i) \leq 2(\lceil 5k_i\log \Gamma \rceil)$. 
\end{claim} 
Using Claim \ref{lem_Hi}, we have 
\begin{align}
\nonumber
\sum\limits_{i=0}^P \alpha_{CF}(\mathcal{H}_i) 
&\leq 2\sum\limits_{i=0}^P \lceil 5k_i\log \Gamma \rceil 
\leq 10\log \Gamma\sum\limits_{i=0}^P k_i + 2P + 2 \\
\nonumber
&\leq 10\log \Gamma\sum\limits_{i \geq 0} \frac{\ka}{2^i}+ 2P + 2
\\
\label{eqn1001}&\leq  20\ka\log \Gamma + 2P + 2=O(\log^2\Gamma).
\end{align}
Now using Observation \ref{obv_union}, we have that $\alpha_{CF}(\mathcal{H}) \leq \alpha_{CF}(\mathcal{G}) + \sum\limits_{i=0}^P \alpha_{CF}(\mathcal{H}_i)$. Using (\ref{eqn1001}) now, the proof is complete.
\end{IEEEproof}
Using Theorem \ref{thm_CF_cover} in conjunction with Lemma \ref{lemma:upperboundCFcover}, we have the following achievability result for the PICOD problem. The achievability of lengths $m,n$ are trivial consequences of the problem setup. 
\begin{corollary}
\label{corollary_achievability}
For any $(n,m,{\mathfrak I})$-PICOD problem, let $\Gamma=\max_{r\in[n]}|\{r'\in[n]\backslash r:I_r\cap I_{r'}\neq \emptyset\}|.$ Then there exists a binary linear scalar PIC for the given problem with length $O(\min\{\log^2{\Gamma},m\}).$ Thus $\ell_1^*({\hcal})=O(\min\{\log^2{\Gamma},m\}).$
\end{corollary}
\subsubsection*{Comparison with known achievability results}
The original work of Brahma and Fragouli \cite{brahma_fragouli_PICOD} showed the existence of an achievable scheme with length $O(\min\{\log m(1+\log^+(\frac{n}{
\log m})), m, n\})$ (where $\log^+(x)=\max\{0,\log(x)\}$). For $m=n^\delta$ for some $\delta>0,$ this means the existence of a PIC with length $O(\log^2 n)$ is guaranteed.  Our result, Theorem \ref{thm_CF_cover}, gives an upper bound based on the parameter $\Gamma$ of the hypergraph. Given the set of vertices $V$ and edges $\cal E$ of a hypergraph, the parameter $\Gamma$ can be determined in $O(|V||{\cal E}|^2)$ time by a simple algorithm which runs through each edge computing its intersection with all other edges. Further the parameter $\Gamma \leq |{\cal E}|-1=n-1$ always, but it could be much smaller in general, as suggested by the below example. 
\begin{example}
Consider the hypergraph $\mathcal{H} = (V, \mathcal{E})$, where $V = [m]$, $\mathcal{E} = \{\{i,i+1,i+2\}:i\in [m-2]\}$ for $m\geq 3.$ Since every hyperedge overlaps with at most $3$ other edges, we have $\Gamma=3.$  The result from \cite{brahma_fragouli_PICOD} suggests the existence of a code of length $O(\log^2m)$, where by Theorem \ref{thm_CF_cover}, we have a code of constant length (as $m$ grows).     
\end{example}


In \cite{PolyTime_PICOD} an achievable scheme was presented for a PICOD problem with $n$ receivers with length $O(\log^2n).$ The algorithm in \cite{PolyTime_PICOD} had running time polynomial in the problem parameters $m,n$. \rogers{Our result also yields a polynomial time algorithm. 
However, the algorithm does not follow immediately from the proof.
The main difficulty in getting a deterministic algorithm is the presence of the Local Lemma 
in the proof.  
Derandomization of the Local Lemma to provide an constructive algorithm has been studied \cite{lllcgh, harris2019deterministic}. Applying Theorem 1.1 (1) in  \cite{harris2019deterministic}, we get a conflict-free coloring of a hypergraph using $O(t\Gamma^{\frac{1+\epsilon}{t}}\log \Gamma)$ colors, where $t$ and $\Gamma$ are as defined in Theorem \ref{thm_pach_main} and $\epsilon> 0 $ is a constant. This suffices to get a deterministic polynomial time coloring algorithm for the hypergraph $\mathcal{G}$ in the proof of Theorem \ref{thm_CF_cover} using $O(\log^2 \Gamma)$ colors. In a similar way, one can get polynomial time algorithms for constructing conflict-free collection of colorings for hypergraphs $\mathcal{H}_i$ in the proof such that the total number of colors used across all the colorings in such a collection is $O(k_i\log \Gamma)$.}

\begin{remark}
    Similar to Remark \ref{rem:randalg1ambda} and Remark \ref{rem:end}, we can use Theorem
    \ref{thm:moser} to obtain easy-to-implement randomized polynomial time algorithm for constructing scalar
    PICODs of length $O(\log^2 \Gamma)$.
\end{remark}

\section{`Local' Conflict-Free Chromatic Number and Pliable Index Coding}
\label{sec:local}
In this section, we define the \textit{local} conflict-free chromatic number. This results in a refined upper bound for $\ell^*_1(\hcal)$ which can be smaller than $\alpha_{CF}(\hcal)$. 

\begin{definition}[Local Conflict-Free Chromatic Number] 
\label{def:localcfchrom}
Given a hypergraph ${\hcal}(V,{\cal E})$, the local conflict-free chromatic number of 
${\hcal}$ is given by
$$ \Delta({\hcal}) = \min_{\substack{C : C \textrm{ is a CF}\\ \textrm{coloring of }{\hcal}}} \underbrace{\max_{E \in {\cal E}} \; \lvert \{C(v) : v \in E\}\rvert}_{\Delta_C(\hcal)}.$$

For convenience, we define $\Delta_C(\hcal) = \max_{E \in {\cal E}} \; \lvert \{C(v) : v \in E \}\rvert$, where 
 $C$ is a
 conflict-free coloring of ${\hcal}$. 
Therefore, 
$\Delta({\hcal}) = \min_C \Delta_C(\hcal)$, where the minimum is over all such  colorings $C$ of ${\hcal}$.
\end{definition}

We have the following observation.

\begin{observation}
\label{obv:DeltaXcomparison-new}
$\Delta({\hcal})\leq  \xcal_{CF}({\hcal}).$
\end{observation}

The following lemma shows that the gaps between $\Delta({\hcal}),$ $\alpha_{CF}({\hcal}),$ and $\xcal_{CF}({\hcal}),$ can be quite large. 
\begin{lemma}
\label{lemma_2uniformparameters2}
There exists a hypergraph $\hcal$ with $n$ hyperedges for which $\Delta({\hcal})=2,$ while $\alpha_{CF}({\hcal})=\Theta(\log n)$ and $\xcal_{CF}({\hcal})=\Theta(\sqrt{n})$.
\end{lemma}
\begin{IEEEproof}
Consider the $2$-uniform hypergraph with $m$ vertices and all the $2$-sized subsets of $[m]$ as hyperedges. We have already shown the values of the parameters $\alpha_{CF}({\hcal})$ and $\xcal_{CF}({\hcal})$ in Lemma \ref{lemma_2uniformparameters1}. 
Since every hyperedge is of size 2, $\Delta({\hcal}) \leq 2$.
To see that $\Delta({\hcal})=2,$ we first observe that in any conflict-free coloring (which uses $m$ colors), at least $m-1$ colors are present in any essential color set. Thus, for $m\geq 3$ there exists at least one hyperedge containing $2$ essential colors. Hence $\Delta({\hcal})=2.$  
\end{IEEEproof}

We now show that there is a scalar PIC of length $\Delta({\hcal})$ for the PICOD problem given by ${\hcal}([m],{\mathfrak I})$, provided we are operating over a sufficiently large finite field. For $K,N$ being positive integers such that $K\leq N$, let a linear code of dimension $K$ and length $N$ be referred to as an $[N,K]$ code. Recall that in the $K\times N$ generator matrix of a maximum distance separable (MDS) $[N,K]$ code, any $K$ columns are linearly independent. Below we define the MDS matrix associated with a given coloring $C$ of the graph ${\hcal}([m],{\mathfrak I})$.  
\begin{definition}[{MDS matrix associated with a CF coloring of ${\hcal}$}]
\label{defn:MDSmatrixcoloring}
 Let $C$ be a conflict-free coloring of ${\hcal}([m],{\mathfrak I})$ that uses colors from $[D]$.
 Let $\Delta_C=\max_{E \in {\cal E}} \; \lvert \{C(v) : v \in E\}\rvert$.
 Let $G'$ denote the $\Delta_C\times D$ generator matrix of a $[D,\Delta_C]$ MDS code.  We index the columns of $G'$ by the set $[D]$, and denote the column indexed by $d\in [D]$ as $G'_d$. Consider the $\Delta_C \times m$ matrix $G$ defined as follows (the columns of $G$ are indexed as $G_{i}:i\in[m]$): For $i\in[m]$, we set $G_{i}=G'_{C(i)}.$

We refer to $G$ as the \textit{MDS matrix associated with the coloring $C$}.
\end{definition}

Note that such a matrix defined above always exists when the field size is not smaller than $D$. Using the matrix defined above, we show in the below theorem the achievability of length $\Delta({\hcal})$.

\begin{theorem}
\label{theoremMDS}
$\ell^*_1({\hcal})\leq \Delta({\hcal}).$
\end{theorem}
\begin{IEEEproof}
Refer Appendix \ref{appendix:prooftheoremMDS}. 
\end{IEEEproof}

\begin{example}
We give an example of the code construction involved in proof of Theorem \ref{theoremMDS}. Consider the $2$-uniform hypergraph ${\hcal}$ of $V=\{1,\hdots,10\}$ with all the $\binom{10}{2}$ hyperedges, and let $C$ be a conflict-free coloring of ${\hcal}$. 
Any such coloring requires at least $10$ colors, and there exists an edge with $2$ colors. Thus we have $\Delta({\hcal})=2.$ We define the encoder matrix as $G$ as the encoder matrix of a $[10,2]$ MDS code. It is easy to check that this satisfies all the receivers. 
\end{example}

\subsection{Local conflict-free covering number and PICOD}
\label{subsec:localcovering_PICOD}
In this subsection we define a \textit{local} version of the covering number arising due to conflict-free collections of ${\hcal}$ and relate it to an achievable scheme for the PICOD problem. \prasad{OCTOBER 2022: This quantity is likely to be the strictly the tightest upper bound for $\ell_1^*({\hcal})$ presented in this work, provided the question in red on the next page is answered.}

In the below definition, we will make use of the notion of ``Conflict-free collection of colorings'', defined in Definition \ref{defn:cfcoveringnumber}.

\begin{definition}[Local Conflict-Free Covering Number] 
\label{def:localcfcovering}
Given a hypergraph ${\hcal}(V,{\cal E})$, the local conflict-free covering number of 
${\hcal}$ is given by
$$ \lambda({\hcal}) = \min_{\substack{\mathfrak{C} : \mathfrak{C} \textrm{ is a CF}\\ \textrm{collection of }{\hcal}}} 
\overbrace{
\sum_{C^p\in \mathfrak{C}}
\underbrace{\left (
\max_{E \in {\cal E}} \; \lvert \{C^p(v) : v \in E\}\rvert
\right )}_{\Delta_{C^p}(\hcal)}}^{\lambda_{\mathfrak{C}}(\hcal)} 
.$$

For convenience, we define $\lambda_{\mathfrak{C}}(\hcal) =\sum_{C^p\in \mathfrak{C}}  \Delta_{C^p}(\hcal) 
$, where 
 $\mathfrak{C}$ is a
 conflict-free collection of colorings of ${\hcal}$. 
Therefore, 
$\lambda({\hcal}) = \min_{\mathfrak C} \lambda_{\mathfrak C}(\hcal)$, where the minimum is over all such  collections $\mathfrak C$ of ${\hcal}$.
\end{definition}

 We now show that there is a achievable PICOD scheme for the PICOD hypergraph ${\hcal}([m],{\mathfrak{I}})$ with length $\lambda({\hcal}).$
\begin{theorem}
\label{thm:localcoverbound}
$\ell_1^*({\hcal})\leq \lambda({\hcal}).$
\end{theorem}
\begin{IEEEproof}
Let $\mathfrak{C}=\{C^1,\hdots,C^P\}$ be a conflict-free collection of colorings of ${\hcal}([m],{\mathfrak{I}})$.
For convenience, we use $\Delta(p)$ to denote $\Delta_{C^p}({\hcal})$. We want to show a valid PIC encoder matrix for ${\hcal}$ of size $\sum_{p\in [P]}\Delta(p)$. Invoking the definition of $\lambda({\hcal}),$ the proof is complete. 

Let the MDS matrix of size $\Delta(p)\times m$ associated with the coloring $C^p$ be denoted by $G^p$. Consider the $\sum_{p\in[P]}\Delta(p)\times m$ matrix 
\begin{align*}
G=\begin{bmatrix}G^1\\\vdots\\G^P\end{bmatrix}.
\end{align*}
We now show that the matrix $G$ is a valid encoder of a PIC for ${\hcal}.$ By definition of $\mathfrak{C}$, for any edge $I_r\in\mathfrak{I},$ there is some  coloring $C^p$ in which there exists some $v\in I_r$ such that $v$ is colored by $C^p$
and $C^p(v)\cap C^p(i)=\emptyset, \forall i\in I_r\setminus \{v\}.$ By arguments similar to the proof of Theorem \ref{theoremMDS}, we have that $G^p$ satisfies receiver $r$. As $r$ is arbitrary, by Lemma \ref{lemma:concatenatedGeneratorsPIC}, $G$ satisfies all receivers and is a valid PIC. 
\end{IEEEproof}

Also, the following observation is clearly true by definition, and because any conflict-free coloring $C$ of ${\hcal}$ generates a conflict-free collection of ${\hcal}$ containing only $C$. 
\begin{observation}
\label{obv_relation_btwn_LCF_and_LCFCover}
$\lambda({\hcal})\leq \min(\alpha({\hcal}),\Delta({\hcal}))$.
\end{observation}

\section{The $t$-requests case}
\label{sec:trequest}
In the previous sections, we considered the PICOD setting with each receiver demanding one of the $m$ messages. In the present section, we generalize our results to the scenario where each receiver $r$ has to be sent any $\min(|I_r|,t)$-sized subset of messages indexed by its request set $I_r$. We shall call PICOD schemes which satisfy the above $t$-requests scenario as $t$-\textit{request pliable index codes}, or $t$-request PICs. In the rest of the section, we shall denote the smallest length of any $t$-request PIC for the PICOD problem defined by hypergraph $\hcal$ as $\ell^{*(t)}(\hcal)$. It was shown in \cite{PolyTime_PICOD} that, for any PICOD problem with $n$ receivers, a $t$-request PIC with length $O(t\log n+\log^2 n)$ exists, and can be designed in polynomial time (in number of receivers $n$ and the messages $m$). 

The notion of strong conflict-free coloring of hypergraphs was introduced by Horev et al. in \cite{Elad_2010_StrongCFcoloring} as a conflict-free coloring in which any edge of the hypergraph `sees' more than one distinct color. Formally, a $t$-\textit{strong conflict-free coloring} of hypergraph ${\hcal}=(V,{\cal E})$ with $L$ labels (or colors) is an assignment $C:V\rightarrow [L]$ such that the following holds. 
\begin{itemize}
    \item For any edge $E\in {\cal E}$, there exist $\min(t,|E|)$ vertices in $E$ which get distinct labels, i.e., there exists $V_E\subseteq E$ such that (a) $|V_E|=\min(t,|E|)$, (b) $|\{C(v):v\in V_E\}|=|V_E|$, and (c) $\{C(v):v\in V_E\}\cap \{C(v'):v'\in E\setminus V_E\}=\emptyset$.
\end{itemize}
The minimum $L$ such that a $t$-strong conflict-free coloring exists for ${\hcal}$ is then called the 
\emph{$t$-strong conflict-free chromatic number} of ${\hcal}$, which we denote by ${\xcal}^{(t)}({\hcal})$. The notion of a \textit{$t$-strong conflict-free collection} of colorings and the $t$\textit{-strong conflict-free covering number} of $\hcal$ (denoted by $\alpha^{(t)}(\hcal)$) are defined as in Definition \ref{defn:cfcoveringnumber}, with the only difference being that we want the colorings $C^p:p\in[P]$ to be $t$-strong conflict-free colorings for the subgraphs $\hcal_p:p\in[P]$, respectively. 
Similarly, we can define the \textit{t-strong local conflict-free chromatic number} as follows. 
\begin{align*}
    \Delta^{(t)}(\hcal)\triangleq \min\{\Delta_C(\hcal):C~\text{is a}~t \text{-strong conflict-free coloring for}~\hcal\},
\end{align*}
where $\Delta_C(\hcal)$ is as in Definition \ref{def:localcfchrom}. Finally, the $t$\textit{-strong local conflict-free covering number} $\lambda^{(t)}({\hcal})$ can be defined as 
\begin{align*}
\lambda^{(t)}({\hcal})\triangleq\min\{ \lambda_{\mathfrak C}(\hcal):\mathfrak C \mbox{ is a } t\mbox{-strong conflict-free collection of }\hcal\},
\end{align*}
where $\lambda_{\mathfrak{C}}({\hcal})$ is as in Definition \ref{def:localcfcovering}. 
The following results capture the utility of $t$-strong conflict-free colorings to the $t$-request PICOD problems. 
\begin{lemma}
\label{lemma:tstrong_coloringlemma}
Let $C$ be $t$-strong conflict-free coloring for a PICOD hypergraph $\hcal$ that uses $L$ colors. Then there exists a $t$-request PIC of length $L$ for $\hcal$ over every field. Further, there is a $t$-request PIC for $\hcal$ of length $\Delta_{C}(\hcal)$ over any field $\Fb$ with $|\Fb|\geq L$
where $\Delta_{C}(\hcal)$ is as in Definition \ref{def:localcfchrom}. 
\end{lemma}
\begin{IEEEproof}
The proof for the first part follows that of Lemma \ref{lemma:upperboundchromnumber}, while the proof for the second part follows that of Theorem \ref{theoremMDS} via MDS matrices defined in Definition \ref{defn:MDSmatrixcoloring}.

\textcolor{red}{WHERE EXACTLY ARE WE REQUIRING THAT $\Fb$ with $|\Fb|\geq L$. It does not seem to be mentioned in the MDS section, i.e., Def \ref{defn:MDSmatrixcoloring} or in the discussion above it.}
\end{IEEEproof}
\begin{lemma}
\label{lemma:tstrong_coveringlemma}
Let ${\mathfrak C}=\{C^p:p\in[P]\}$ be a $t$-strong conflict-free collection for a PICOD hypergraph $\hcal$, that uses $L_p:p\in [P]$ colors for the colorings $C^p:p\in[P]$ respectively. Then there is a $t$-request PIC  of length $\sum_{p\in [P]}L_p$ for $\hcal$ over every field. 
Further, there is a $t$-request PIC of length $\lambda_{\mathfrak{C}}(\hcal)$ for $\hcal$ over every field $\Fb$ such that $|\Fb|\geq \sum_{p\in [P]}L_p$, where $\lambda_{\mathfrak{C}}(\hcal)$ is as defined in Definition \ref{def:localcfcovering}.
\end{lemma}
\begin{IEEEproof}
The first part follows by arguments similar to Lemma \ref{lemma:upperboundCFcover}, while the second part uses arguments as in Theorem \ref{thm:localcoverbound}.
\end{IEEEproof}
We then have the following theorem which summarizes the extensions of our results to the $t$-request PICOD scenario.
\begin{theorem}
\label{thm:t-request_coloring_picod_relation}
For the PICOD problem defined by $\hcal$, we have,
\begin{align*}
    \ell^{*(t)}(\hcal)\leq \lambda^{(t)}(\hcal)\leq \min(\Delta^{(t)}(\hcal),\alpha^{(t)}(\hcal))\leq \xcal^{(t)}(\hcal).
\end{align*}
\end{theorem}
\begin{proof}
The claim that $\lambda^{(t)}(\hcal)$ is an upper bound for $\ell^{*(t)}(\hcal)$ follows from Lemma \ref{lemma:tstrong_coveringlemma}.
The proofs for the claim $\lambda^{(t)}(\hcal)\leq\min(\Delta^{(t)}(\hcal),\alpha^{(t)}(\hcal))$ follows similar to Observation \ref{obv_relation_btwn_LCF_and_LCFCover}. The last inequality follow by definition of the quantities involved. 
\end{proof}

\subsection{Upper Bound for $\lambda^{(t)}(\hcal)$}
\label{sec:lambdaupperboundtstrong}

Before we prove bounds on the $t$-strong conflict-free chromatic numbers, 
we prove a partitioning lemma. 

\begin{lemma}
\label{lem_partitioning_hypergraph}
Let $\mathcal{H}=(V,\mathcal{E})$ be a hypergraph with $|V| = m$ and $|\mathcal{E}| = n$. It is given that every hyperedge in $\mathcal{H}$ intersects at most $\Gamma$ other hyperedges. Let $\ell$ be a positive integer such that $\ell \geq \log (6(\Gamma+1))$. Then, there exist $V_1, \ldots , V_r \subseteq V$ and $\mathcal{E}_1 \uplus \mathcal{E}_2 \uplus \cdots \uplus \mathcal{E}_r \uplus \mathcal{E}' = \mathcal{E}$ with $r < \log_2m$ such that 
\\(i) $\forall i \in [r],~(E \in \mathcal{E}_i) \implies (6\ell < |E \cap V_i| \leq 36\ell)$, and
\\(ii) $(E \in \mathcal{E}') \implies (|E| \leq 12\ell)$. 
\end{lemma}
\begin{proof}
Let $\mathcal{E}' = \{E \in \mathcal{E}~:~|E| \leq 12 \ell\}$. Let $r$ be the \rogers{largest} integer so that $12\ell < \frac{m}{2^{r-1}}$. Thus, $r < \log_2m$. For each $i \in [r]$, let $m_i = \frac{m}{2^i}$ and let $\mathcal{E}_i = \{E \in \mathcal{E}\setminus \mathcal{E}'~:~m_i < |E| \leq m_{i-1}\}$. Consider an $i \in [r]$. Below we explain how we construct $V_i$. Independently and uniformly at random select a vertex $v \in V$ into $V_i$ with probability $\frac{12\ell}{m_i}$. Let $X_E^i$ be a random variable that denotes $|E \cap V_i|$, for a  hyperedge $E \in \mathcal{E}_i$. Let $\mu_E^i := E[X_E^i]$. Then, $\mu_E^i = \frac{12|E|\ell}{m_i}$. Since $m_i < |E| \leq m_{i-1}$, we have $12 \ell < \mu_E^i \leq 24\ell$. 
Applying the Chernoff bound given in Theorem \ref{thm_Chernoff} with $\delta = 1/2$, we get $Pr[|X_E^i - \mu_E^i| \geq \frac{1}{2}\mu_E^i] \leq 2e^{-\frac{\mu_E^i}{12}} \rogers{<} 2e^{-\frac{12 \ell}{12}} \leq 2e^{-\log(6(\Gamma+1))} = \frac{1}{3(\Gamma+1)}$. 

Let $A_E^i$ denote the bad event that $|X_E^i - \mu_E^i| \geq \frac{\mu_E^i}{2}$. We have shown that $Pr[A_E^i] \leq  \frac{1}{3(\Gamma+1)}$. Since $e\cdot \frac{1}{3(\Gamma+1)}\cdot (\Gamma+1) \leq 1$, 
from Lemma \ref{lem:local} we get $Pr[\wedge_{E \in \mathcal{E}_i} (\lnot A_E^i)] > 0$. Hence, there exists a $V_i$ such that $\forall E \in \mathcal{E}_i$,  $|X_E^i - \mu_E^i| < \frac{1}{2}\mu_E^i$. Since $12\ell < \mu_E^i \leq 24 \ell$, this implies that there is a $V_i$ such that $\forall E \in \mathcal{E}_i$, $6\ell < |E \cap V_i| \leq 36\ell$. 
\end{proof}

Using Theorem \ref{thm:moser}, we obtain the below algorithm for constructing the partitions.

\begin{mybox}{Partition($V, \mathcal E_i, \ell$)}

\textbf{Input:} (i) A set $V$, (ii) $\mathcal E_i = \{ E \subseteq V \;:\; m_i \leq |E| \leq m_{i-1}\}$, where $m_i = m/2^i$ and every
set in $\mathcal E_i$ overlaps with at most $\Gamma$ other sets in $\mathcal E_i$, and (iii) a positive integer $\ell > \log (6 (\Gamma+1))$.

\textbf{Output:} A set $V_i \subseteq V$ such that every set $E \in \mathcal E_i$ satisfies the condition $6 \ell < |E \cap V_i| \leq 36 \ell$.

\textbf{Algorithm:}
\begin{itemize}
	\item Each vertex $v \in V$ is independently and uniformly chosen to be in $V_i$ with probability $\frac{12\ell}{m_i}$.
	\item While $\exists E \in \mathcal E_i$ that does not satisfy the output condition
	
	\begin{itemize}
		\item Choose an arbitrary $E \in \mathcal E_i$ that does not satisfy the output condition
		\item Resample each vertex $v \in E$. That is, for each $v \in E$, independently and uniformly decide to include 
		it in $V_i$ with probability  $\frac{12\ell}{m_i}$.
	\end{itemize}
	\item Output $V_i$.
\end{itemize} 

\end{mybox}

\begin{remark} \label{rem:part}
We can apply Theorem \ref{thm:moser} by considering the sampling of the vertices as the random variables, 
and the events 
$A_E^i$  as the bad events.
So we 
need at most $|\mathcal E_i|/\Gamma \leq n/\Gamma$ resamplings in expectation. Prior to each resampling, we would also need
to test if all $E \in \mathcal E_i$ satisfies the output condition for the current choice of $V_i$. This takes $O(mn)$ time. So overall 
the time required is at most $O(mn^2/\Gamma)$ in expectation.
\end{remark}

\begin{theorem}
\label{thm_local_kCF_cover}
Let $\mathcal{H} = (V,\mathcal{E})$ be a hypergraph with $|V| = m$ and $|\mathcal{E}| = n$. It is given that every hyperedge in $\mathcal{H}$ intersects at most $\Gamma$ other hyperedges. 
Then, for any positive integer $t$, 
$$\lambda^{(t)}(\mathcal{H}) = \max(O(\log \Gamma \log m), O(t \log m)). $$ 
\end{theorem}
\begin{proof}
Let $\ell = \max(\log (6(\Gamma+1)), t)$. Then, from Lemma \ref{lem_partitioning_hypergraph}, we have $V_1, \ldots , \rogers{V_{r}} \subseteq V$ and $\mathcal{E}_1 \uplus \cdots \uplus \mathcal{E}_r \uplus \mathcal{E}'= \mathcal{E}$ with $r < \log_2 m$ such that (i) $\forall i \in [r],~(E \in \mathcal{E}_i) \implies (6\ell < |E \cap V_i| \leq 36\ell)$, and
(ii) $(E \in \mathcal{E}') \implies (|E| \leq 12\ell)$. For each $i \in [r]$, we define hypergraphs $\mathcal{H}_i = (V, \mathcal{E}_i)$. We define a local $t$-strong conflict-free coloring $c_i$ for $\mathcal{H}_i$ using $|V_i| + 1$ colors in which all the vertices in $V_i$ get a distinct color from the first $|V_i|$ colors and all the vertices in $V\setminus V_i$ get the $(|V_i| + 1)^{th}$ color. In such a coloring, each hyperedge $E \in \mathcal{E}_i$ sees at least $6\ell + 1$ colors exactly once and at most $36\ell + 1$ colors in total. Finally, we define a local $t$-strong conflict-free coloring  $c'$ for $\mathcal{H}'$ using $|V|$ colors that gives a distinct color to every vertex in $V$. In this coloring, every hyperedge $E \in \mathcal{E}'$ sees $|E|$ (which is $\leq 12\ell$) colors, each color exactly once. This completes the proof of the theorem. 
\end{proof}

\begin{remark} 
\label{rem:randalg1ambda}
By Remark \ref{rem:part}, we have a randomized algorithm for the partitioning in Lemma \ref{lem_partitioning_hypergraph}
that runs in expected time is at most $O(mn^2/\Gamma)$. Thus we have an algorithmic version of the above construction
that runs in expected time at most $O(mn^2\log m/\Gamma)$.
\end{remark}

\subsection{Upper Bound for $\alpha^{(t)}(\hcal)$}
\label{sec:alphaupperboundtstrong}

\begin{theorem}
\label{thm_kCF_cover}
Let $\mathcal{H} = (V,\mathcal{E})$ be a hypergraph with $|V| = m$ and $|\mathcal{E}| = n$. It is given that every hyperedge in $\mathcal{H}$ intersects at most $\Gamma$ other hyperedges. Let $t_1, t$ be two positive integers with $t_1 = \max(\log(6(\Gamma+1)), t)$. It is given that $\forall E \in \mathcal{E}$, $|E| > 12t_1$. Then, 
$$\alpha^{(t)}(\mathcal H) = \max(O(\log \Gamma \log m), O(t\log m)).$$
\end{theorem}
\begin{proof}
Applying Lemma \ref{lem_partitioning_hypergraph} with $\ell = t_1$, we get $V_1, \ldots , V_r \subseteq V$ and $\mathcal{E} _1 \uplus \cdots \uplus \mathcal{E}_r = \mathcal{E}$ with $r < \log_2 m$ such that $\forall i \in [r],~(E \in \mathcal{E}_i) \implies (6t_1 < |E \cap V_i| \leq 36t_1)$. Note that $\mathcal{E}' = \emptyset$ as every hyperedge in $\mathcal{H}$ is of size greater than $12t_1$. 

Consider an $i \in [r]$. We describe a $t$-strong conflict-free coloring for the hyperedges in  $\mathcal{E}_i$. For each vertex in $V_i$, assign a color that is chosen independently and uniformly at random from a set of $19 e^2 t_1$ colors. All the vertices in $V \setminus V_i$ are assigned the same color, a color different from the $19e^2 t_1$ colors used to color the vertices in $V_i$. 
For each $E \in \mathcal{E}_i$, let $z_E^i \triangleq |E \cap V_i|$. We have $6t_1 < z_E^i \leq 36t_1$. Let $B_E^i$
be the bad event that $E\cap V_i$ is colored with $\leq \lceil \frac{z_E^i+ t}{2} \rceil \leq \frac{z_E^i + t + 1}{2}$ colors, where the last inequality holds since $z_E^i$ and $t$ are integers. 
Note that if $B_E^i$ does not occur, then $E \cap V_i$ has some $t$ colors that appear
exactly once. Now we estimate the probability of $B_E^i$.

\begin{eqnarray*}
    Pr[B_E^i] & \leq & \binom{19 e^2 t_1}{(z_E^i+t+1)/2} \left( \frac{(z_E^i + t + 1)/2}{19 e^2 t_1}\right)^{z_E^i} \\
    & \leq & \left(\frac{e \cdot 19e^2 t_1}{(z_E^i+t+1)/2} \right)^{(z_E^i+t+1)/2}
   \left( \frac{(z_E^i + t + 1)/2}{19 e^2 t_1}\right)^{z_E^i} \quad \quad (\mbox{since  }\binom{n}{k} \leq \left(\frac{en}{k}\right)^k) \\
    & = & e^{t+1} \left(\frac{ z_E^i + t + 1}{38 e t_1}\right)^{(z_E^i - t- 1)/2}
    \\
    & \leq & e^{t_1 + 1} \left(\frac{ 36 t_1 + t_1 + 1}{38 e t_1}\right)^{(6t_1 + 1 - t_1- 1)/2} \quad (\mbox{since  } t \leq t_1, \mbox{ and } 6t_1 < z_E^i \leq 36t_1)
    \\
    & \leq & e^{t_1 + 1} \left( \frac{1}{e} \right)^{5t_1/2}  \leq \frac{e}{e^{1.5 t_1}} \leq \frac{e}{e^{1.5 \log (6 (\Gamma + 1))}}\\
    & < &\frac{1}{4 (\Gamma+1)} \quad . 
\end{eqnarray*}

Since $e\cdot \frac{1}{4(\Gamma+1)}\cdot (\Gamma +1) \leq 1$, by Lemma \ref{lem:local}, we get $Pr[\wedge_{E \in \mathcal{E}_i}(\lnot B_E^i)] > 0$. That is, $\forall E \in \mathcal E_i$, there is a 
$t$-strong conflict free coloring of $E \cap V_i$ with $19 e^2 t_1$ colors.
This completes the proof of the theorem. 
\end{proof}

\begin{mybox}{$t$-Strong-CF-Covering($\mathcal H = (V,  \mathcal E), t$)}

\textbf{Input:} (i) A hypergraph $\mathcal H = (V, \mathcal E)$, with $|V| = m$ and $|\mathcal E| = n$. 
			Every hyperedge $E \in \mathcal E$ intersects with at most $\Gamma$ other hyperedges $E' \in \mathcal E$.
			(ii) A positive integer $t$. Set $t_1 = \max(\log (6 (\Gamma+1)), t)$. Every $E \in \mathcal E$ satisfies $|E| > 12t_1$.
			
\textbf{Output:} A conflict-free collection of colorings $\mathcal C = \{c_1, c_2, \ldots, c_r\}$ of $\mathcal H$, where each 
$c_i: V \rightarrow \{0\} \cup [19 e^2 t_1]$ and $r < \log_2 m$.

\textbf{Algorithm:}
\begin{itemize}
	\item For $i = 1$ to $\log_2 m - 1$
	\begin{itemize}
		\item $\mathcal E_i = \{ E \in \mathcal E \;:\; m_i < |E| \leq m_{i-1}\}$, where $m_i = m/2^i$.
		\item $V_i = \mbox{Partition}(V, \mathcal E_i, t_1)$.
		\item For every $v \in V \setminus V_i$, $c_i(v) = 0$.
		\item For each vertex $v \in V_i$, let $c_i(v)$ be a color chosen independently and uniformly at random from $[19 e^2 t_1]$.
		\item While $\exists E \in \mathcal E_i$ such that $E \cap V_i$ contains  $\leq  \lceil \frac{z_E^i+ t}{2} \rceil$ distinct colors
	\begin{itemize}
		\item Choose an arbitrary $E \in \mathcal E_i$ such that $E \cap V_i$ contains  $\leq  \lceil \frac{z_E^i+ t}{2} \rceil $ distinct colors.
		\item Recolor each vertex $v \in E$. That is, for each $v \in E$, independently and uniformly assign $c_i(v)$ from $[19 e^2 t_1]$.
	\end{itemize}
	\end{itemize}
	\item Output $\mathcal C = \{c_1, c_2, \ldots, c_r\}$.
\end{itemize} 

\end{mybox}


\begin{remark} \label{rem:end}
For each $i$, we have a randomized algorithm for the partitioning in Lemma \ref{lem_partitioning_hypergraph}
that runs in expected time  at most $O(mn^2/\Gamma)$ (see Remark \ref{rem:part}).

We can again apply Theorem \ref{thm:moser} by considering the coloring of the vertices as the random variables, and the events
$B_E^i$  as the bad events. So we 
need at most $|\mathcal E_i|/\Gamma \leq n/\Gamma$ recolorings in expectation. Prior to each recoloring, we would also need
to test if any of the $B_E^i$'s occur for the current coloring. This takes $O(mn)$ time. So overall 
the time required is at most $O(mn^2/\Gamma)$ in expectation.

Since $i$ ranges from 1 to $\log_2 m - 1$, the total running time is $O(mn^2\log m/\Gamma)$ in expectation.
\end{remark}

\subsection{A $t$-request instance that requires $\Omega(t \log m/\log t)$ length PICOD}
\label{sec:tstrongtight}

For a PICOD hypergraph $\hcal = (V, \mathcal E)$ with $|V| = m$ and $|E | = n$,
Theorem \ref{thm_local_kCF_cover}
combined with Theorem \ref{thm:t-request_coloring_picod_relation} gives us a $t$-request PIC of length $\max(O(\log \Gamma \log m), O(t \log m))$.
However, this PIC is based on an MDS code and is defined over a large field. 
Theorem \ref{thm_kCF_cover} combined with Theorem \ref{thm:t-request_coloring_picod_relation} gives us a $t$-request PIC of length
$\max(O(\log \Gamma \log m), O(t \log m))$. This is not based on an MDS code. Hence the field can be of a smaller size. 
But the code yielded by Theorem \ref{thm_kCF_cover} works only when all the hyperedges (request sets) are sufficiently large.
In this section, we demonstrate a $t$-request instance that requires $\Omega(t \log m)$ length PICOD.

Consider the following hypergraph $\hcal = (V, \mathcal E)$ where $V = [m]$ and $\mathcal E = \mathcal E_1 \uplus \mathcal E_2 \uplus \ldots \uplus \mathcal E_r$,
where $r = \lfloor (\log m/ \log (12t)) \rfloor - 1$. We have $\mathcal E_1 = \{\{1, 2, \ldots, m\}\}$, and for $2 \leq i \leq r$, we have,
$\mathcal E_i = \left \{ \{1, 2, \ldots, \frac{m}{(12t)^{i-1}}\}, \{\frac{m}{(12t)^{i-1}} + 1, \ldots, \frac{2m}{(12t)^{i-1}}\}, \ldots, \{(m- \frac{m}{(12t)^{i-1}} + 1, \ldots, m\} \right\}$.
Notice that $(12t)^{r+1} \leq m \leq (12t)^{r+2}$, by our choice of $r$. 

The total number of hyperedges is given by 
$$\mathcal E = 1 + (12t) + \ldots + (12t)^{r-1} = \frac{(12t)^{r-1} -1}{12t - 1} < \frac{m}{11t}. $$
Hence we also have the overlapping parameter $\Gamma \leq |\mathcal E| < m/(11t)$.

Let $t$ be an integer such that $t > \log m$. Since $\Gamma < m/(11t)$, we have the following: 
$$\log(6 \Gamma + 1) < \log (6m/(11t) + 1) < \log m < t.$$

That is $\max(\log (6 \Gamma + 1), t) = t$. Since the hyperedges in $\mathcal E_r$ are of size $m/(12t)^{r-1} > 12t$, we 
satisfy all the conditions of Theorem \ref{thm_kCF_cover}.

We use the following result from \cite{Sowjanya}. We can see that the hypergraph $\hcal$ constructed above satisfies the 
conditions in the theorem below. Using the below theorem, we get that $\ell^{*(t)}(\hcal) \geq tr = t \left (\lfloor (\log m/ \log (12t)) \rfloor - 1\right ) = 
\Omega(t \log m/\log t)$. The implies that the length of the PIC schemes given by Theorems \ref{thm_local_kCF_cover} and 
\ref{thm_kCF_cover} are asymptotically tight up to a multiplicative factor of $\log t$.

\begin{theorem}
    Consider a PICOD hypergraph $\hcal = (V, \mathcal E)$ corresponding to the $t$-requests case. Suppose there exists a collection
    of subsets of $\mathcal E$, given by $\{\mathcal E_i \subseteq \mathcal E: i \in [r]\}$, such that the following condition holds
    for each $i \leq r-1$: For each $E \in \mathcal E_i$, and for any subset $T \subseteq E$ with $|T| = t$, there exists an edge
    $E' \in \mathcal E_{i+1}$, such that $T \cap E' = \emptyset$.

    Then, $\ell^{*(t)}(\hcal) \geq t r$.
\end{theorem}

\section{Extension to $k$-vector PIC}
\label{sec:k-folddefinition}
In this section, we briefly show that the idea of using conflict-free colorings for constructing scalar PICs extend naturally to $k$-vector PICs as well. Towards this end, we define the notion of 
$k$-fold conflict-free coloring of a hypergraph, which generalizes the definition of a conflict-free coloring. To the best of our knowledge this generalized notion is not available in literature. 
\begin{definition}
\label{definitionkfold}
A $k$-fold coloring of a hypergraph ${\hcal}=(V,{\cal E})$ is an assignment of $k$-sized subsets of $[L]$ 
to the vertices $V$, given by $C:V\rightarrow \binom{[L]}{k}$. A $k$-fold coloring is \textit{conflict-free for edge} $E\in{\cal E}$ if there exists some $v\in E$ such that $C(v)\cap C(v')=\emptyset$, for each $v'\in E\backslash v.$ A coloring $C$ is a \textit{$k$-fold conflict-free coloring} for ${\hcal}$ if $C$ is a $k$-fold conflict-free coloring for each edge in ${\cal E}.$ We define the  \textit{$k$-fold conflict-free chromatic number} of ${\hcal}$ as the smallest $L$ such that a $k$-fold conflict-free coloring of ${\hcal}$ exists as defined above, and denote it by $\xcal_{k,CF}({\hcal})$. 
\end{definition}
\begin{figure*} 
    \centering
  \subfloat[\label{1a}]{%
       \includegraphics[width=0.25\linewidth]{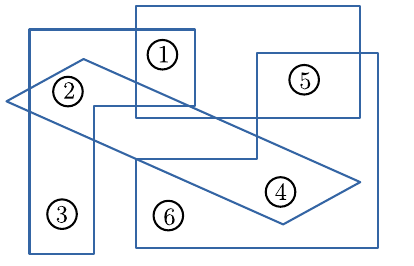}}
       \hfill
  \subfloat[\label{1b}]{%
        \includegraphics[width=0.25\linewidth]{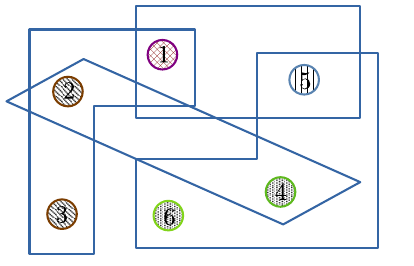}}
      \hfill
  \subfloat[\label{1c}]{%
        \includegraphics[width=0.25\linewidth]{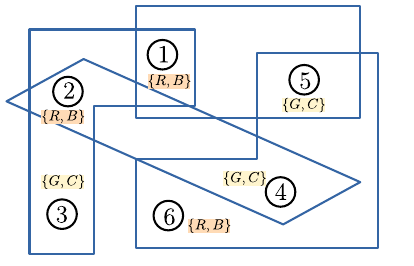}}
  \caption{\small \textcolor{black}{Figure (a) shows a hypergraph ${\hcal}$ with $6$ vertices and edge set ${\cal E}=\{\{1,2,3\},\{1,5\},\{2,4\},\{4,5,6\}\}$. Figure (b) represents a $1$-fold conflict-free coloring with $4$ colors, with the color classes $\{1\},\{2,3\},\{5\},\{4,6\}$. Figure (c) shows a $2$-fold conflict-free coloring using the colors $\{R,G,B,C\}.$}}
  \label{fig1}
  \vspace{-0.6cm}
\end{figure*}
Observe that $\xcal_{1,CF}({\hcal})=\xcal_{CF}({\hcal}).$ Fig. \ref{fig1} gives an example of $1$-fold and $2$-fold conflict-free coloring. Clearly, $\xcal_{k,CF}({\hcal})\leq k\xcal_{CF}({\hcal})$ as we can always obtain a $k$-fold conflict-free coloring from a $1$-fold conflict-free coloring by expanding each color into $k$ unique colors. However we show an example here for which this inequality is strict.
\begin{example}
\label{example_kfoldand1foldcomparison}
Consider the hypergraph $\hcal$ given by vertex set $V=\{a,\hdots,e\}$ and ${\cal E}=\left\{\{a,b\},\{b,c \},\right.$ $\left.\{c,d\},\{d,e\},\{e,a\}\right\}.$ 
Consider any $1$-fold coloring of this graph. 
It is easy to see that two colors are not sufficient to give a 1-fold conflict-free coloring.
It is also easy to find a  conflict-free coloring with $3$ colors, for instance, give color $1$ to vertices $\{a,c\}$, color $2$ to $\{b,d\}$ and color $3$ to vertex $e$. Thus $\xcal_{1,CF}(\hcal)=3.$	

Similarly, we can show that there cannot be a $2$-fold conflict-free coloring with $4$ colors. Now consider the following $2$-fold coloring with $5$ colors denoted by $\{1,\hdots,5\}$. Let set $\{1,2\}$ be assigned  to vertex $a$, $\{3,4\}$ to $b$, $\{5,1\}$ to $c$, $\{2,3\}$ to $d$ and $\{4,5\}$ to $e$. It is easy to check that this is a $2$-fold conflict-free coloring. Thus $\xcal_{2,CF}({\hcal})=5<6=2(\xcal_{1,CF}({\hcal})).$
\end{example}
We now define the indicator matrix of $k$-fold coloring of $\hcal$, which leads to a $k$-vector PIC achievability scheme for $\hcal$. 
\begin{definition}
\label{def:kfold-ind}
Let $C:V\rightarrow \binom{[L]}{k}$ denote a $k$-fold coloring of ${\hcal}(V,{\mathfrak{I}})$. Let $C(i)=\{C_{i,1},\hdots,C_{i,k}\}$ denote the subset assigned to the vertex $i\in V =[m]$. Consider a standard basis of the $L$-dimensional vector space over ${\mathbb F}$, denoted by $\{e_1,\hdots,e_L\}$. Now consider the $L\times mk$ matrix $G$ (with columns indexed as $\{G_{i,j}:i\in[m],j\in[k]\}$) constructed as follows. 
\begin{itemize}
    \item For each $i\in[m],j\in[k]$, column $G_{i,j}$ of $G$ is fixed to be $e_{C_{i,j}}$.
\end{itemize}
We call $G$ as the \textit{indicator matrix associated with the coloring $C$. }
\end{definition}
Using Definition \ref{def:kfold-ind}, the following achievability result naturally follows. 
\begin{lemma}
\label{lemma:upperboundfoldchromnumber}
$\ell^*_k({\hcal})\leq \xcal_{k,CF}({\hcal})$.
\end{lemma}
\begin{IEEEproof}
Let $G$ denote the indicator matrix associated with a $k$-fold conflict-free coloring $C$ as defined in Definition \ref{def:kfold-ind}. Let $C(i)=\{C_{i,1},\hdots,C_{i,k}\}$ be the set assigned to vertex $i$.

Since $C$ is conflict-free, any edge $I_r$ of ${\hcal}$ has a vertex $d$ such that $C(d)\cap C(i)=\emptyset, \forall i\in I_r\setminus d$.  Then, we have $\{e_{C_{d,j}}:j\in[k]\}\cap \{e_{C_{i,j}}:j\in[k]\}=\emptyset,$ for any $i\in I_r\backslash d.$ This also means $span(\{e_{C_{d,j}}: j\in[k]\})\cap span(\{e_{C_{i,j}}: i\in I_r\backslash d, j\in[k]\})=\{\boldsymbol{0}\}$, 
since the vectors involved are basis vectors. 
Further, as $|\{C_{d,j}:j\in[k]\}|=k,$ hence $\dim(span(\{e_{C_{d,j}}: j\in[k]\}))=k.$ Thus, $G$ satisfies receiver $r$ by Lemma \ref{lemmaindependence}. As $r$ is arbitrary, $G$ is a valid encoder for $\hcal$. By definition of $\xcal_{k,CF}({\hcal}),$ the proof is complete. 
\end{IEEEproof}
The conflict-free collection can similarly be generalized from Definition \ref{defn:cfcoveringnumber}. 
\begin{definition}[\textit{$k$-fold conflict-free collection, $k$-fold conflict-free covering number}]
\label{kfoldcfcoveringnumber}
Let $\mathcal{H}=(V,\mathcal{E})$ be a hypergraph. Let $\mathfrak{C} = \{C^1, \ldots , C^P\}$ where each $C^p:V\rightarrow \binom{[L_p]}{k}$ is a   $k$-fold colorings of the hypergraph $\hcal$. We say $\mathfrak{C}$ is a \emph{conflict-free collection of $k$-fold colorings of $\mathcal{H}$}, if 
there exists a collection of $P$ subgraphs $\hcal_p:p\in[P]$ such that  $\hcal=\cup_{p\in[P]}\hcal_p$ and the coloring $C^p$ (when restricted to vertices in $\hcal_p$) is a $k$-fold conflict-free coloring for $\hcal_p$.

The minimum value of the sum $\sum_{p=1}^PL_p$ over all possible collections $\mathfrak{C}$ (over all $P$) as defined above, is called the \emph{$k$-fold conflict-free covering number of $\mathcal{H}$} denoted by $\alpha_{k,CF}(\mathcal{H})$.

\end{definition}
Note that $\alpha_{1,CF}({\hcal})=\alpha_{CF}({\hcal})$. 
The following definition helps us to define the notion of $k$-fold conflict-free colorings. 
\begin{definition}[Local $k$-fold conflict-free chromatic number]
\label{defn:kfold_local}
Given a hypergraph ${\hcal}(V,{\cal E})$, the local $k$-fold conflict-free chromatic number of 
${\hcal}$ is given by
$$ \Delta_k({\hcal}) = \min_{\substack{C : C \textrm{ is a k-fold CF}\\ \textrm{coloring of }{\hcal}}} \underbrace{\max_{E \in {\cal E}} \; \lvert \bigcup_{v\in E} C(v)\rvert}_{\Delta_{k,C}(\hcal)}.$$

For convenience, we define $\Delta_{k,C}(\hcal) = \max_{E \in {\cal E}} \; \lvert \bigcup_{v\in E} C(v) \rvert$, where 
 $C$ is a $k$-fold
 conflict-free coloring of ${\hcal}$. 
Therefore, 
$\Delta_k({\hcal}) = \min_C \Delta_{k,C}(\hcal)$, where the minimum is over all such  
$k$-fold colorings $C$ of ${\hcal}$. 
\end{definition}
Similarly, we can define the local $k$-fold conflict-free covering number of $\hcal$, denoted by $\lambda_k({\hcal})$, extending Definition \ref{def:localcfcovering}.

Using similar arguments as in Section \ref{sec:conflictfreecolorings}, we get the following
bounds. 
\begin{lemma}
$\ell^*_k(\mathcal{H})\stackrel{(a)}{\leq} \lambda_k({\hcal})\stackrel{(b)}{\leq}\min(\Delta_k({\hcal}), \alpha_{k,CF}(\mathcal{H})) \stackrel{(c)}{\leq}\xcal_{k,CF}(\mathcal{H})$. Further, $\beta_k(\hcal)\leq k\beta_1(\hcal)$ where $\beta_k$ refers to any of the parameters $\lambda_k,\Delta_k,\alpha_{k,CF},\xcal_{k,CF}$. 
\end{lemma}
\begin{IEEEproof}
The proof for (a), (b), and (c) follow similar to the proofs of Theorem \ref{thm:localcoverbound}, Observations \ref{obv_relation_btwn_LCF_and_LCFCover} and  \ref{obv:DeltaXcomparison-new}, and Lemma \ref{lemma:generalhypergraph_alphabounds}.

The proof for the second part follows by the observation that expanding each coloring in a conflict-free coloring (or a collection) of $\hcal$ to $k$-unique colors results in a $k$-fold conflict-free coloring (or, respectively, a collection) of $\hcal$. 
\end{IEEEproof}
\section{Discussion}
\label{sec:discussion}
We have presented a  hypergraph coloring framework for the pliable index coding problem.
We provide easy-to-implement randomized algorithms for the PICOD problem and the $t$-request PICOD 
problem. These algorithms can be derandomized using existing techniques. However, such 
deterministic algorithms may be cumbersome to implement.
It would be interesting to give simpler deterministic polynomial-time algorithms for the same. 
Explicit algorithms for $k$-vector pliable index coding which give non-trivial improvements over simple extensions of scalar index codes would certainly be interesting. Finally more investigation is needed into the gaps between the  parameters presented in this work. 

\bibliographystyle{IEEEtran}
\bibliography{CF_PICOD.bib}
\appendices
\section{Tools from probability}

Below we state the Local Lemma, due to Erd\H{o}s and Lovasz, which is required in some of our proofs.

\begin{lemma}[\emph{The Local Lemma}, \cite{lovaszlocallemma}] \label{lem:local} Let $A_1, \ldots , A_n$ be events in an arbitrary probability space. Suppose that each event $A_i$ is mutually independent of a set of all the other events $A_j$ but at most $d$, and that $Pr[A_i] \leq p$ for all $i \in [n]$. If $ep(d+1) \leq 1,$ then $Pr[\cap _{i=1}^n \overline{A_i}] > 0$.  
\end{lemma}

Moser and Tardos \cite{mosertardos} demonstrated an algorithmic version of the Local Lemma. They showed the following\footnote{The paper \cite{mosertardos} states the result and algorithm in a general setting. To avoid clutter, we state a specific symmetric case which suffices our requirements.}.

\begin{theorem}[Algorithmic Local Lemma \cite{mosertardos}] 
\label{thm:moser}
Let $\mathcal P$ be a finite set of mutually independent random variables in a 
probability space. Let $A_1, \ldots, A_n$ be events that are determined by these variables. 
Suppose that each event $A_i$ is mutually independent of a set of all the other events $A_j$ but at most $d$, and that $Pr[A_i] \leq p$ for all $i \in [n]$. If $$ep(d+1) \leq 1, $$
then there exists an assignment of the random variables in $\mathcal P$ such that none of the events $A_i$ occur. Moreover, there is a randomized algorithm (described below) that finds such an assignment, that uses at most $n/d$ resampling steps in expectation.
\end{theorem}

\begin{mybox}{Algorithmic Local Lemma}
First, sample the random variables in $\mathcal P$ as per the distribution. If at least one of the events $A_i$ occur, choose an arbitrary $A_j$ that occurs. 
Then resample the random variables in $\mathcal P$ that determine $A_j$. Repeat this until none of the events $A_i$ occur.
\end{mybox}

A version of the Chernoff bound is stated below.
\begin{theorem}[Chernoff Bound, Corollary 4.6 in \cite{mitzenmacher}]
\label{thm_Chernoff}
Let $X_1, \ldots , X_n$ be independent Poisson trials such that $Pr[X_i] = p_i$. Let $X = \sum_{i=1}^n X_i$ and $\mu = E[X]$. For $0 < \delta < 1$, 
$$ Pr[|X-\mu| \geq \delta \mu] \leq 2e^{-\mu \delta^2/3}. $$ 
\end{theorem}

\section{Proof of Claim \ref{lem_Hi}}
\label{appendix_proof_claim}
We wish to show $\alpha_{CF}(\mathcal{H}_i) \leq 2(\lceil 5k_i\log \Gamma \rceil)$. 
Let $q_i = \frac{1}{k_i}$ and $t_i = \lceil 5k_i\log \Gamma \rceil$. We do $t_i$ rounds of coloring of the vertex set $V$ of $\mathcal{H}_i$, using two new colors in each round. 
In any given round, 
we color each vertex $v \in V$ independently with probability $q_i$ with first color, and give it the second color with the remaining
probability, i.e., $1- q_i$.

Consider a hyperedge $E \in \mathcal{E}_i$. Let $F_E$ denote the `bad' event that none of the $t_i$ colorings of $V$ is a conflict-free coloring for $E$. The probability that the coloring in a given round is a conflict-free coloring for $E$ is at least $|E|q_i(1-q_i)^{|E|-1}$. Thus, 
\begin{eqnarray*}
Pr[F_E] &\leq & \left(1 - |E|q_i(1-q_i)^{|E|-1}\right)^{t_i} \\
& \stackrel{(a)}{\leq} & \frac{1}{e^{t_i|E|q_i(1-q_i)^{|E|-1}}} \\
& \stackrel{(b)}{\leq}  & \frac{1}{e^{t_i\frac{k_i}{2}\frac{1}{k_i}(1-\frac{1}{k_i})^{k_i-1}}}\\  
& \stackrel{(c)}{\leq} & \frac{1}{e^{\frac{t_i}{2k_i}}}\\
& \leq & \frac{1}{\Gamma^{2.5}}~~(\mbox{since }t_i \geq 5k_i\log \Gamma ),
\end{eqnarray*}
where $(a)$ holds by inequality $1+x\leq e^x,$ $(b)$ holds as $ q_i = \frac{1}{k_i}$ and $~\frac{k_i}{2} \leq |E| < k_i$, $(c)$ holds using the inequality $(1+x)^r \geq 1 + rx~\mbox{for } x \geq -1, r \in \mathbb{R}\setminus (0,1).$ 
Thus, for each hyperedge $E$ in $\mathcal{H}_i$, the probability of the bad event $F_E$ is at most $p := \frac{1}{\Gamma^{2.5}}$. Observe that each such event $F_E$ is independent of all the other events, but at most $d := \Gamma$ events corresponding to those hyperedges intersecting with $E$. Since $ep(d+1) \leq 1$, by Local Lemma (Lemma \ref{lem:local}), $Pr[\cap_{E \in \mathcal{E}_i}\overline{F}_E] > 0$. This proves the lemma. 

\section{Proof of Theorem \ref{theoremMDS}}
\label{appendix:prooftheoremMDS}
Let $C$ be a conflict-free coloring of ${\hcal}$, that uses colors
from the set $[D]$.
We now show that the MDS matrix $G$ associated with the coloring $C$ (of size $\Delta_C(\hcal) \times m$ as defined in Definition \ref{defn:MDSmatrixcoloring}) satisfies the properties in Lemma \ref{lemmaindependence}, and hence is a valid PIC for ${\hcal}$. Note that this code would have length $\Delta_C(\hcal)$, which we denote by simply $\Delta_C$. By definition of $\Delta({\hcal}),$ our proof would then be complete. 

To see this, consider any $I_r\in{\mathfrak I}$. By definition of $G$ and $\Delta_C$, we have that \begin{align}
    \label{eqn1}
|\{G_{i}:i\in I_r\}|\leq \Delta_C.
\end{align}

Now, by the definition of conflict-free coloring $C$, at least one vertex $v\in I_r$ is such that $C(v)\neq C(v'), \forall v'\in I_r\backslash v.$ Thus, by the definition of matrix $G$, we have the following.

We first note that the vector $G_v$ appears exactly once in the collection $\{G_{i}:i\in I_r\}$. 
Since the columns of $G$ are taken from the generator matrix $G'$ of a $[D,\Delta_C]$ MDS code, any $\Delta_C$ distinct columns of $G$ are linearly independent. 
Further, $G_v$ is linearly independent from the space spanned by any collection of $(\Delta_C-1)$ other columns of $G'$. By (\ref{eqn1}) and the above observation, the columns in $\{G_i:i\in I_r\backslash v\}$ lie 
in the subspace spanned by the $(\Delta_C-1)$ columns of $G'$ apart from $\{G_{v}\}$. Thus, we have that  $span(\{G_{v}\})\cap span(\{G_{i}:i\in I_r\backslash v\})=\{\boldsymbol{0}\}$. Thus $G$ satisfies receiver $r$ by Definition \ref{DefnGeneratorForEdge}. As $r$ is arbitrary, by Lemma \ref{lemmaindependence}, $G$ represents a valid PIC for ${\hcal}$.
\end{document}